\documentclass[runningheads, 10pt]{llncs}
\usepackage[T1]{fontenc}
\usepackage{caption}
\usepackage{subcaption}
\usepackage{graphicx} 
\usepackage{algorithm}
\usepackage{algorithmic}
\usepackage{amssymb}
\usepackage{mathtools}

\usepackage{tikz} 
\usetikzlibrary{positioning}
\usetikzlibrary{arrows}
\usetikzlibrary{arrows.meta}
\usepackage{pgfplots} 

\usepackage{graphicx}
\begin{document}
\title{Fair Healthcare Rationing to Maximize Dynamic Utilities}
%
%
\author{Aadityan Ganesh\inst{1}$^*$ \and
    Pratik Ghosal\inst{1}$^*$ \and
    Vishwa Prakash HV\inst{1}$^*$\and
    Prajakta Nimbhorkar\inst{1,2}$^*$}
%
 \institute{Chennai Mathematical Institute, India \and
UMI ReLaX \\
\email{\{aadityanganesh,pratik,vishwa,prajakta\}@cmi.ac.in}}
\maketitle              

\def\thefootnote{*}\footnotetext{The authors contributed equally to this work and are listed in alphabetical order}\def\thefootnote{\arabic{footnote}}

\begin{abstract}
Allocation of scarce healthcare resources under limited logistic and infrastructural facilities is a major issue in the modern society. We consider the problem of allocation of healthcare resources like vaccines to people or hospital beds to patients in an online manner. Our model takes into account the arrival of resources on a day-to-day basis, different categories of agents, the possible unavailability of agents on certain days, and the utility associated with each allotment as well as its variation over time. 

We propose a model where priorities for various categories are modelled in terms of utilities of agents.
We give online and offline algorithms to compute an allocation that respects eligibility of agents into different categories, and incentivizes agents not to hide their eligibility for some category. The offline algorithm gives an optimal allocation while the online algorithm gives an approximation to the optimal allocation in terms of total utility. Our algorithms are efficient, and maintain fairness among different categories of agents. Our models have applications in other areas like refugee settlement and visa allocation. We evaluate the performance of our algorithms on real-life and synthetic datasets. The experimental results show that the online algorithm is fast and performs better than the given theoretical bound in terms of total utility. Moreover, the experimental results confirm that our utility-based model correctly captures the priorities of categories.
\end{abstract}

\section{Introduction}
Healthcare rationing has become an important issue in the world amidst the COVID-19 pandemic.
At certain times the scarcity of medical resources like vaccines, hospital beds, ventilators, medicines especially in developing countries  raised the question of fair and efficient distribution of these resources. 
One natural approach is to define priority groups. For example,
for vaccination, the main priority groups considered include health care workers, workers in other essential services, and people with vulnerable medical conditions \cite{Persad20,Truog20}. Racial equity has been another concern \cite{Bruce21}. Having made the 
priority groups, it still remains a challenge to allocate resources within the groups in a transparent manner \cite{Emanuel20,WHO20}.
A New York Times article has mentioned this as one of the hardest decisions for health organizations \cite{Fink20}. In light of this,
it is a major problem to decide how to allocate medical resources fairly and efficiently while respecting the priority groups and other
ethical concerns.

The healthcare rationing problem has been recently addressed by market designers. In \cite{Pathak20}, the problem was framed as
a two-sided matching problem (see e.g. \cite{Roth90}). Their model has reserve categories each with its own priority ordering of people.
This ordering is based on the policy decisions made according to various ethical guidelines. It is shown in \cite{Pathak20} that running the Deferred Acceptance algorithm of Gale and Shapley \cite{GS62} has desired properties like eligibility compliance, non-wastefulness and respect to priorities. This approach of \cite{Pathak20} has been
recommended or adopted by organizations like the NASEM (National Academies of Sciences, Engineering, and Medicine) \cite{NASEM20}. It has also been recognized in medical literature \cite{Persad20,Sonmez20}, and is mentioned by the Washington Post \cite{WP}. The Smart Reserves algorithm of \cite{Pathak20} gives a maximum matching satisfying the desired properties mentioned earlier. However, it assumes a global priority ordering on people. In a follow-up work, \cite{Aziz21} generalize this to
the case where categories are allowed to have heterogeneous priorities. Their Reverse Rejecting (REV) rule, and its extension to Smart Reverse Rejecting (S-REV) rule are shown to satisfy the goals like eligibility compliance, respect to priorities, maximum size, non-wastefulness, and strategyproofness.

However, the allocation of healthcare resources is an ongoing process. On a day-to-day basis, new units arrive in the market and they
need to be allocated to people. The variation in the availability of medical resources over a period of time, and the possible unavailability of recipients on certain days is an important factor in making allocation decisions. For example, while allocating vaccines, the unavailability of people on certain days might lead to wastage of vaccines, especially if the units are reserved for categories a priori. The previous models do not encompass this dynamic nature of resources. Moreover, the urgency with which a resource needs to be allocated to an individual also changes over time. While priority groups or categories aim to model this by defining a priority order on people, defining a strict ordering is not practically possible. While dealing with a large population, defining a strict ordering on people is not desirable. For instance, in the category of old people, it is neither clear nor desirable to define a strict order on people of the same age 
and same vulnerabilities. Even if categories are allowed to have ties in their ordering, the ordering still provides only an ordinal ranking.

Our model provides the flexibility to have cardinal rankings in terms of prioritizing people by associating a {\em utility value} for each individual. Thus, in our work, categories do not define an ordering over people, instead, there is a utility value associated with allocation of the resource to each person. The goal is to find an allocation with maximum total utility while respecting category quotas. However, utilities can change over time. For instance, the advantage of allotting a ventilator to a person today might be far more beneficial than allotting it tomorrow. Similarly, vaccinating the vulnerable population as early as possible is much more desirable from a social perspective than delaying it to a later day. We model this through {\em dynamic utilities}. Thus, we consider utilities that diminish over time. The {\em discounting factor}  $0<\delta<1$ is multiplicative. Such exponential discounting is commonly used in economics literature \cite{Ramsey1928,Samuelson1937}. 
Our utility maximization objective can thus be seen as maximization of {\em social welfare}. Another advantage is that the division of available units into various categories is not static. It is dynamically decided depending on the supply of units and availability of people on each day. 

Our algorithms to find a maximum utility allocation are based on network flows. They adhere to the following important ethical principles which were introduced by Aziz et al in \cite{Aziz21}:
\begin{enumerate}
    \item complies with the eligibility requirements 
    \item is strategyproof (does not incentivize agents to under-report the categories they qualify for or days what they are available on),
    \item is non-wasteful (no unit is unused but could be used by some eligible agent)
\end{enumerate}

Additionally our algorithms give an approximate maximum weight matchings, where the weights denote the utility value of a matching. We note that the current state-of-practice algorithms such as first-come first-serve or random ordering do not guarantee non-wastefullness as the matching returned by them may not be of maximum size. Furthermore, matchings returned by these algorithms could be of arbitrarily low total utility. Using category quotas and utility values, we provide a framework in which more vulnerable populations can be prioritized while maintaining a balance among the people vaccinated through each category on a day-to-day basis. 

\subsection{Related Work}
The topic of constrained matching problems has been an active area or research and it has been studied in the context of school choice and hospital residents problem apart from healthcare rationing \cite{Kojima19,Aziz21,Kamada15,Kamada17,Biro10,Goto16,Sankar21}. The setting with two-sided preferences has been considered in \cite{Hamada16,Kavitha14,Huang10,Kamiyama16}. The fairness and welfare objectives have been covered in a comprehensive manner in \cite{Moulin03}. 

Another application of the constrained matching problem is in the refugee resettlement problem. Refugee resettlement is a pressing matter in the twenty-first century where people have been forced to leave their country in order to escape war, persecution, or natural disaster. 
In the refugee resettlement process the refugee families are settled from asylum countries to the host countries where the families are given permanent residentship. The reallocation is done keeping in mind the necessities of the families as well as the infrastructure capacities of the host countries. Delacrétaz et al. \cite{delacretaz2016refugee} formalized refugee allocation as a centralized matching market design problem. 
The refugee allocation problems have been studied both in terms of matching problems with preferences \cite{andersson2016assigning,delacretaz2019matching,aziz2018stability,jones2017international,jones2018local,nguyen2021stability,sayedahmed2020refugee} and without preferences\cite{bansak2018improving,delacretaz2016refugee}. In the matching problem with preferences, the goal is to match the refugees to localities based on the preference of either one or both sides, while satisfying the multidimensional resettlement constraints. Delacretaz et al. considered the problem both in terms of with and without preferences. The problem without preference can be reduced to multiple multidimensional knapsack problems \cite{delacretaz2016refugee}. The branch-and-bound method can be used to find the exact solution. Bansak et al. \cite{bansak2018improving} used a combination of supervised machine learning and optimal matching to obtain a refugee matching that satisfies the constraints of refugees and localities.
The dynamic version of the refugee resettlement problem \cite{andersson2018dynamic,ahani2021dynamic,cilali2021location} has also been considered in literature. 

\subsection{Our Models}
We define our model below and then define its extension. 
Throughout this paper, we consider vaccines as the medical resource to be allocated. People are referred to as agents.
Note that although the discussion assumes perishability 
of resources, it can easily be extended to non-perishable resources. 
\paragraph{Model 1: }\label{model_1}
Our model consists of a set of agents $A$, a set of categories $C$, and a set of days $D$. For day $d_j\in D$, there is a {\em daily supply} $s_j$ denoting the number of vaccine shots available for that day. For each category $c_i\in C$, 
and each day $d_j\in D$, we define a {\em daily quota} $q_{ij}$. 
Here $q_{ij}$ denotes the maximum number of vaccines that can be allocated for $c_i$ on day $d_j$. There is a priority factor $\alpha_k$ associated with an agent $a_k$. Let $\alpha_{\max}~=~\max_i\{\alpha_i~\mid~\alpha_i~\text{ is the priority factor of agent $a_i$}\}$ and $\alpha_{\min}~=~\min_i\{\alpha_i~\mid~\alpha_i~\text{ is the priority factor of agent $a_i$}\}$. Utilities have a {\em discount factor} $\delta\in (0,1)$ denoting the multiplicative factor with which the utilities for vaccinating agents reduce with each passing day. Thus if $a_k$ is vaccinated on day $d_j$, the utility obtained is $\alpha_k\cdotp\delta^j$. Each agent $a_k$ has an {\em availability vector} $v_k\in \{0,1\}^{|D|}$. The $j$th entry of $v_k$ is $1$ if and only if $a_k$ is available for vaccination on day $d_j$.

\paragraph{Model 2:}\label{model_2} Model 2 is an extension of Model 1 in the following way. The sets $A,C,D$ and the daily supply and daily quotas are the same as those in model 1. Apart from the daily quota, each category $c_i$ also has an {\em overall quota} $q_i$ that denotes the maximum total number of vaccines that can be allocated for category $c_i$ over all the days. Note that overall quota is also an essential quantity in 
applications like visa allocation and refugee settlement.

In both the models, a matching $M:A\rightarrow (C\times D)\cup \{\emptyset\}$ is a function denoting the day on which a person is vaccinated and the category through which it is done, such that the category quota(s) and daily supply values do not exceed on any day. Thus if we define variables $x_{ijk}$ such that $x_{ijk}=1$ if $M(a_k)=(c_i,d_j)$ and $x_{ijk}=0$ if $M(a_k)=\emptyset$, then
we have $\sum_{i,j} x_{ijk}\leq 1$ for each $k$, $\sum_{k,j}x_{ijk}\leq q_i$ for each $i$, $\sum_k x_{ijk} \leq q_{ij}$ for each $i,j$, and $\sum_{i,k} x_{ijk}\leq s_j$ for each $j$. Here $1\leq i \leq |C|, 1\leq j\leq |D|, 1\leq k\leq |A|$. If $M(a_k)=\emptyset$ for some $a_k\in A$, it means the person could not be vaccinated through our algorithm within $|D|$ days.

In both the models, the utility associated with $a_k$ is \(\alpha_k\cdotp\delta^{j-1}\) where $M(a_k)=(c_i,d_j)$. The goal is to find a matching that maximizes the total utility.  


\subsection{Our Contributions}
The utilities $\alpha_k$ and discounting factor $\delta$ have some desirable properties. If agent $a_k$ is to be given a higher priority over agent $a_\ell$, then we set $\alpha_k>\alpha_\ell$. On any day $d_j$, $\alpha_k\cdot \delta^j>\alpha_\ell\cdot \delta^j$. Moreover, the difference in the utilities of the two agents diminishes over time i.e. if $j<j'$ then $(\alpha_k-\alpha_\ell)\delta^j >(\alpha_k - \alpha_{\ell})\delta^{j'}$. Thus the utility maximization objective across all days vaccinates $a_k$ earlier than $a_\ell$. 

We consider both online and offline settings. The offline setting relies on the knowledge about availability of agents on all days. This works well in a system
where agents are required to fill up their availability in advance e.g. in case of planned surgeries, and visa allocations. The online setting involves knowing the availability of all the agents only on the current day as in a {\em walk-in} setting. 
Thus the availability of an agent on a day in future is not known.

We give an optimal algorithm for Model $1$ in the offline setting.. 

\begin{theorem}\label{thm:off-opt}
There is a polynomial-time algorithm that computes an optimal solution for any instance of Model $1$ in the offline setting.
\end{theorem}

We also give algorithms for both Model $1$ and Model $2$ in the online setting.
We give theoretical guarantees on the performance of online algorithms in terms of their {\em competitive ratio} in comparison with the utility of an offline optimal solution. 

\begin{theorem}\label{thm:gen-on-same-utility}
There is an online algorithm (Algorithm~\ref{alg:online-greedy-m1}) that gives a competitive ratio of (i) $1+\delta$ for Model $1$ and (ii) of \(1+\delta + ({\alpha_{\max}}/{\alpha_{\min}})\delta\) for Model $2$ when $\delta$ is the common discounting factor for all agents. The algorithm runs in polynomial time.
\end{theorem}
We prove part $(i)$ of Theorem~\ref{thm:gen-on-same-utility} in Section~\ref{sec:analysis-model-1} whereas part $(ii)$ is proved in Appendix.
\paragraph{Strategy-proofness:}
It is a natural question whether agents benefit by hiding their availability on some days. We show that the online algorithm is strategy-proof.
In this context, we analyze our online algorithm for Model $1$ from a game theoretic perspective. We exhibit that the offline setting has a {\em pure Nash equilibrium} that corresponds to the solution output by the online algorithm. For this, we assume that the tie-breaking among agents is done according to an arbitrary permutation $\pi$ of agents. 
\begin{theorem}\label{thm:nash}
Let an offline optimal solution that breaks ties according to a random permutation $\pi$ match agent $a_i$ on day $d_i$. Then for each agent $a_i$, reporting  availability exactly on day $d_i$ (unmatched agents mark all days as unavailable) is a pure Nash equilibrium. Moreover, the Nash equilibrium corresponds to a solution output by the online algorithm.
\end{theorem}

\paragraph{Experimental Results:}
We also give experimental results in Section~\ref{sec:simulation} using real-world datasets. Apart from maximization of utilities, we also consider the number of days taken by the online algorithm for vaccinating high priority people. Our experiments show that
the online algorithm almost matches the offline algorithm in terms of both of these criteria.
\paragraph{Selection of utility values:} An important aspect of our model is that the choice of utility values does not affect the outcome as long as the utility values have the same numerical order as the order of priorities among agents. Thus the output of online as well as offline algorithm remains the same as long as $\alpha_k>\alpha_\ell$ whenever agent $a_k$ has a higher priority over agent $a_\ell$.

\section{Optimal Offline Algorithm for Model $1$}\label{sec:offline-optimal-algorithm}

The problem can be modelled as an instance of the minimum cost flow network. We define the minimum cost flow problem here for completeness. 

\paragraph{The minimum cost flow problem:} The input is a flow network $G=(V,E)$ as a directed graph with node set $V$, edge set $E$, capacities $c_e>0$ and cost $u_e\in \mathbb{R}$ on each edge $e\in E$, and a source $s$ and sink $t$. A flow $f:E\rightarrow \mathbb{R}$ is a valid flow in $G$ if $f(e)\leq c(e)$, and the incoming flow at any node except $s$ and $t$ equals the outgoing flow. The cost of a flow $f(e)$ along an edge $e$ is
$u_e\cdot f(e)$. A minimum cost flow in the network is the one that minimizes the sum of costs of the flow along all edges.

There are polynomial-time algorithms known for the minimum cost flow problem. Also, it is known that if all the capacities are integers, then the optimum flow is an integer. We refer the reader to \cite{Ahuja93} for the details of minimum cost flow. 
\paragraph{Reduction: }
The construction of the flow network is shown in Figure~\ref{fig:flow-nw}. The flow network consists of a source $s$, a sink $t$, nodes for each day $d_j$, each agent $a_k$ and nodes $c_{ij}$ for each $(c_i,d_j)\in C\times D$. Each edge $(s,d_j)$ has capacity $s_j$ denoting the daily supply for day $d_j$, each edge $(d_j,c_{ij})$ has capacity equal to $q_{ij}$, and all other edges have capacity $1$. All the edges are directed. Additionally, each $(c_{ij},a_k)$ edge has cost $-u_k\cdot\delta_k^j$ whereas other edges have cost $0$.

\begin{proof}(of Theorem~\ref{thm:off-opt})
We show that a minimum cost flow $f$ in the flow network corresponds to a maximum utility matching in the given instance. The integrality of minimum cost flow implies that each edge incident on $t$ can have a flow of either $0$ or $1$. For each $k$, if $f(a_k,t)=1$, then 
there is exactly one $c_{ij}$ such that $f(c_{ij},a_k)=1$. Set $M(a_k)=(c_i,d_j)$ in the
corresponding matching $M$. Similarly, for any matching $M$, A corresponding flow  can be shown as follows. If $M(a_k)=(c_i,d_j)$ then set $f(a_k,t)=f(c_{ij},a_k)=1$, and set $f(s,d_j)$ equal to the number of agents vaccinated on day $d_j$, and $f(d_j,c_{ij})$ equal to the number of agents vaccinated on day $d_j$ through category $c_i$. It is clear that this is a valid flow in the network, and the negation of the cost of the flow is the same as the utility of the corresponding matching.
\end{proof}

\begin{figure}
  \centering
  \scalebox{1}{\begin{tikzpicture}[myn/.style={circle,thick,draw,inner sep=.1cm,outer sep=2pt},scale = 1,
      terminal/.style={circle, draw=black, fill=black!10, thick, minimum size=8mm},
      ]

      \node[terminal] (s) at (0,0) {$s$};
      \node[myn] (d1) at (3,2) {$d_1$};
      \node[myn] (d2) at (3,-2) {$d_2$};
      \node[myn] (c11) at (6,4) {$c_{11}$};
      \node[myn] (c12) at (6,2.5) {$c_{12}$};
      \node[myn] (c13) at (6,1) {$c_{13}$};
      \node[myn] (c21) at (6,-1) {$c_{21}$};
      \node[myn] (c22) at (6,-2.5) {$c_{22}$};
      \node[myn] (c23) at (6,-4) {$c_{23}$};
      \node[myn] (a1) at (9,2) {$a_1$};
      \node[myn] (a2) at (9,0) {$a_2$};
      \node[myn] (a3) at (9,-2) {$a_3$};
      \node[terminal] (t) at (12,0) {t};
      \draw[thick, ->] (s)--(d1);
      \draw[thick, ->] (s)--(d2);
      \draw[thick, ->] (d1)--(c11);
      \draw[thick, ->] (d1)--(c12);
      \draw[thick, ->] (d1)--(c13);
      \draw[thick, ->] (d2)--(c21);
      \draw[thick, ->] (d2)--(c22);
      \draw[thick, ->] (d2)--(c23);
      \draw[thick, ->] (c11)--(a1);
      \draw[thick, ->] (c13)--(a1);
      \draw[thick, ->] (c21)--(a1);
      \draw[thick, ->] (c23)--(a1);
      \draw[thick, ->] (c11)--(a2);
      \draw[thick, ->] (c12)--(a2);
      \draw[thick, ->] (c22)--(a3);
      \draw[thick, ->] (c23)--(a3);
      \draw[thick, ->] (a1)--(t);
      \draw[thick, ->] (a2)--(t);
      \draw[thick, ->] (a3)--(t);
    \end{tikzpicture}}
  \caption{Flow network for finding a maximum utility matching in Model 1}
\end{figure}
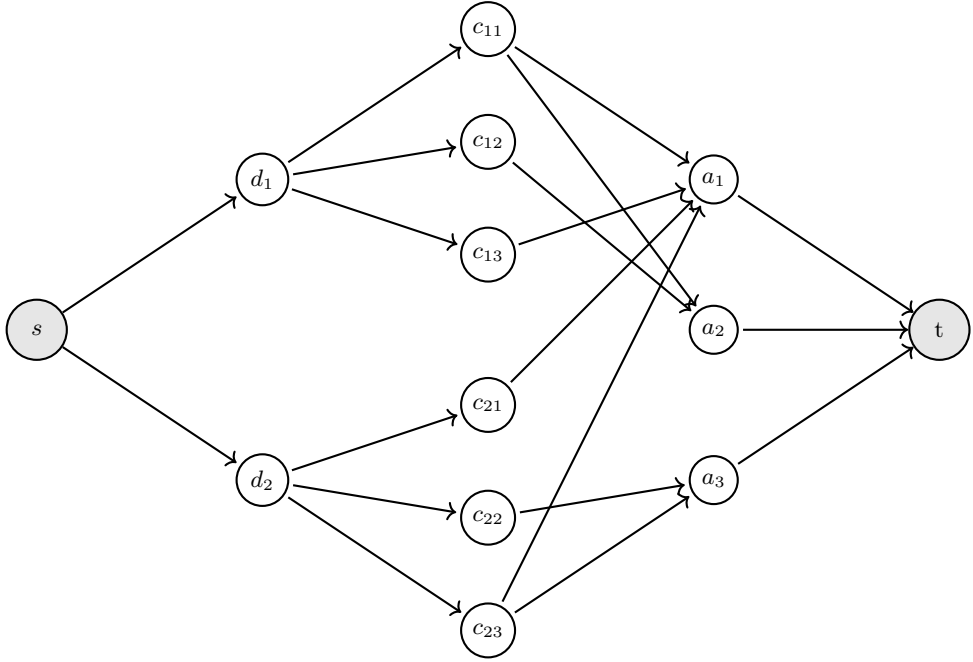\label{fig:flow-nw}

\section{Algorithms for Model~1}\label{sec:model1}
We give a flow based polynomial-time optimal offline algorithm for Model~1 in Appendix. Here, we give an online algorithm for the same which achieves a competitive ratio of \(1+\delta\), where \(\delta\) is the discounting factor of the agents.

\subsection{Online Algorithm for Model~1}
We present an online algorithm which greedily maximizes utility on each day. We show that this algorithm indeed achieves a competitive ratio of \(1+\delta\).\\

\emph{Outline of the Algorithm:} On each day \(d_i\), starting from day \(d_1\), we construct a bipartite graph \({H}_i=({A}_i\cup C, E_i, w_i)\) where \({A}_i\) is the set of agents who are available on day \(d_i\) and are not vaccinated  earlier than day $d_i$. 
Let the weight of the edge \((a_j, c_k) \in E_i\) be \(w_i(a_j, c_k) = \alpha_j.\delta^{i-1}\). We define capacity of the category $c _k \in C$ as  \(b'_{i,k}\). In this graph, our algorithm finds a maximum weighted b-matching of size not more than the daily supply value \(s_i\). 

\begin{algorithm}
  \caption{Online Algorithm for Vaccine Allocation}
  \label{alg:online-greedy-m1}
\textbf{Input:} An instance \(I\) of Model~1 \\
\textbf{Output:} A matching $M:A\to(C\times D)\cup\{\varnothing\}$ \\
\begin{algorithmic}[1] 
    \STATE Let \(D,A,C\) be the set of Days, Agents and Categories respectively.
    \STATE \(M(a_j)\gets \varnothing\) for each $a_j \in A$
    \FOR {day \(d_i\) in \(D\)}
    \STATE \({A}_i \gets \{ a_j \in A \mid a_j\) is available on \(d_i\) and \(a_j\) is not vaccinated\}
    \STATE \(E_i\gets\{(a_j,c_k) \in {A}_i \times C \mid a_j\)is eligible to be vaccinated under category $c_k$ \}
    \FOR  {\((a_j,c_k)\) in \(E_i\)}
    \STATE Let \(w_i(a_j,c_k) \gets \alpha_j\delta^{i-1}\) 
    \ENDFOR
    \STATE Construct weighted bipartite graph \({H}_i=({A}_i\cup C, E_i,w_i)\).
    \FOR {\(c_k\) in \(C\)}
    \STATE \(b'_{i,k} \gets q_{ik}\) \COMMENT {Where \(q_{ik}\) is the daily quota}
    \ENDFOR
    \STATE Find maximum weight b-matching \(M_i\)  in \(H_i\) of size at most \(s_i\). \COMMENT{Where \(s_i\) is the daily supply}
    \FOR {each edge \((a_j,c_k)\) in \(M_i\)}
    \STATE \(M(a_j) \gets (c_k,d_i)\) \COMMENT{Mark \(a_j\) as vaccinated on day \(d_i\) under category \(c_k\)}
    \ENDFOR
    \ENDFOR\\
    \RETURN \(M\)
\end{algorithmic}
\end{algorithm}

The following lemma shows that the maximum weight b-matching computed in Algorithm~\ref{alg:online-greedy-m1} is also a maximum size b-matching of size at most \(s_i\). 

\begin{lemma}\label{lemma:max-size-max-wt}
  The maximum weight b-matching in \(H_i\) of size at most \(s_i\) is also a maximum size b-matching of size at most \(s_i\). 
\end{lemma}

\begin{proof}
  We prove that applying an augmenting path in \(H_i\) increases the weight of the matching. Consider a matching $M_i$ in $H_i$ such that \(M_i\) is not of maximum size and \(|M_i|<s_i\). Let $\rho = (a_1,c_1,a_2,c_2,\cdots,a_k,c_k)$ be an $M_i$-augmenting path in $H_i$. We know that every edge incident to an agent has the same weight in $H_i$. If we apply the augmenting path $\rho$, the weight of the matching increases by the weight of the edge \((a_1,c_1)\). This proves that a maximum weight matching in \(H_i\) of size at most $s_i$ is also a maximum size b-matching of size at most $s_i$.
\end{proof}

\subsection{Charging scheme}\label{sec:analysis-model-1}
We compare the solution obtained by Algorithm~\ref{alg:online-greedy-m1} with the optimal offline solution to get the worst-case competitive ratio for Algorithm~\ref{alg:online-greedy-m1}. Let $M$ be the output of Algorithm~\ref{alg:online-greedy-m1} and $N$ be an optimal offline solution. To compare $M$ and $N$, we devise a {\em charging scheme} by which, each agent $a_p$ matched in $N$ {\em charges} a unique agent $a_q$ matched in $M$. The amount charged, referred to as the {\em charging factor} here is the ratio of utilities obtained by matching $a_p$ and $a_q$ in $M$ and $N$ respectively.

Properties of the charging scheme:
\begin{enumerate}
\item Each agent matched in $N$ charges exactly one agent matched in $M$,
\item Each agent \(a_q\) matched in $M$ is charged by at most two agents matched in $N$, with charging factors at most $1$ and $\delta$. This implies that the utility of $N$ is at most $(1+\delta)$ times the utility of $M$.
\end{enumerate}

We divide the agents matched in $N$ into two types. Type $1$ agents are those which are matched in $M$ on an earlier day compared to that in $N$. Thus $a_p \in A$ is a Type $1$ agent if $a_p$ is matched on day $d_i$ in $M$ and on day $d_j$ in $N$, such that $i<j$. The remaining agents are called Type $2$ agents.
Our charging scheme is as follows:

\begin{enumerate}
\item Each Type~$1$ agent \(a_p\) charges themselves with a charging factor $\delta$, since the utility associated with them in $N$ is at most $\delta$ times that in $M$. 
  
\item Here onwards, we consider only Type $2$ agents and discuss the charging scheme associated with them.

  Let $X_i$ be the set of Type $2$ agents matched on day $d_i$ in $N$, and let $Y_i$ be the set of agents matched on day $d_i$ in $M$. Since Algorithm~\ref{alg:online-greedy-m1} greedily finds a maximum size b-matching of size at most \(s_i\), and as each edge in the b-matching corresponds to a unique agent, we show the following lemma holds:
\end{enumerate}
  \begin{lemma} For each \(d_i \in D\),  the set \(|X_i| \le |Y_i|\).\end{lemma} 

\begin{proof}
  Since $X_i$ contains only Type $2$ agents matched in $N_i$, the agents in $X_i$ are not matched by $M$ until day $i-1$. Therefore $X_i \subseteq A_i$, where $A_i$ is defined in Algorithm~\ref{alg:online-greedy-m1}. 
  The daily quota and the daily supply available for computation of $N_i$ and $M_i$ is the same i.e. $q_{i,k}$, and $s_i$ respectively.
  By construction, $M_i$ is a matching that matches maximum number of agents in $A_i$, up to an upper limit of $s_i$, $|X_i| \leq |Y_i|$. 
\end{proof}
  To obtain the desired competitive ratio we design an injective mapping according to which, each agent \(a_p\) in \(X_i\) can uniquely charge an agent \(a_q\) in \(Y_i\) such that \(\alpha_p \le \alpha_q\). The following lemma shows that such an injective mapping always exists.

  \begin{lemma}\label{lemma:injection-exists}
    There exists an injective mapping \(f:X_i\to Y_i\) such that if \(f(a_p) = a_q\), then \(\alpha_p \le \alpha_q\). 
  \end{lemma}

\begin{proof}
    Let $N_i$ and $M_i$ respectively be the restrictions of $N$ and $M$ to day $d_i$. We construct an auxiliary bipartite graph $G_i$ where $X_i\cup Y_i$ form one bipartition and categories form another bipartition. The edge set is $N_i\cup M_i$. Then we set the capacity of $c_k$ in $G_i$ to be $b_{i,k}=q_{i,k}$. 

    The charging scheme is as follows. Consider the symmetric difference $M_i\oplus N_i$. It is known that  $M_i\oplus N_i$ can be decomposed into edge disjoint paths and even cycles \cite{NasreR17}.
    
    Consider a component \(C\) which is an even cycle as shown in Fig~\ref{fig:cycle-charging}. Since each agent \(a_p\) in \(C\) has both \(M_i\) edge and \(N_i\) edge indecent on it, agent \(a_p\) in \(X_i\) charges her own image in \(Y_i\) with a charging factor of \(1\).

    \begin{figure}[h]
    \begin{subfigure}[b]{0.49\textwidth}
      \centering
      \scalebox{0.5}{
      \begin{tikzpicture}[scale=1]
        
        \node at (0,-0.4) {\(a_1\)};
        \node at (1,-0.4) {\(c_1\)};
        \node at (2.4,1) {\(a_2\)};
        \node at (2.4,2) {\(c_2\)};
        \node at (1,3.4) {\(a_3\)};
        \node at (0,3.4) {\(c_3\)};
        \node at (-1.4,2) {\(a_4\)};
        \node at (-1.4,1) {\(c_4\)};
        \draw[very thick,red] (0,0) -- (1,0);
        \draw[very thick,red] (2,1) -- (2,2);
        \draw[very thick,red] (1,3) -- (0,3);
        \draw[very thick,red] (-1,2) -- (-1,1);

        \draw[very thick,blue] (1,0) -- (2,1);
        \draw[very thick,blue] (2,2) -- (1,3);
        \draw[very thick,blue] (0,3) -- (-1,2);
        \draw[very thick,blue] (-1,1) -- (0,0);

        \path[->,>={Stealth[scale=1.2]}, every loop/.style={min distance=2cm, looseness=40}, thin] (0,0)  edge  [in=-140,out=-40,loop] node [label=below:{1}] {} (0,0);
        
        \path[->,>={Stealth[scale=1.2]}, every loop/.style={min distance=2cm, looseness=40}, thin] (2,1)  edge  [in=-60,out=40,loop] node [label=below:{1}] {} (2,1);
        
        \path[->,>={Stealth[scale=1.2]}, every loop/.style={min distance=2cm, looseness=40}, thin] (1,3)  edge  [in=120,out=40,loop] node [label=above:{1}] {} (1,3);
        
        \path[->,>={Stealth[scale=1.2]}, every loop/.style={min distance=2cm, looseness=40}, thin] (-1,2)  edge  [in=-140,out=130,loop] node [label=below left:{1}] {} (-1,2);

        \draw[fill=black] (0,0) circle (2pt);
        \draw[fill=black] (1,0) circle (2pt);
        \draw[fill=black] (2,1) circle (2pt);
        \draw[fill=black] (2,2) circle (2pt);
        \draw[fill=black] (1,3) circle (2pt);
        \draw[fill=black] (0,3) circle (2pt);
        \draw[fill=black] (-1,2) circle (2pt);
        \draw[fill=black] (-1,1) circle (2pt);

      \end{tikzpicture}}
      \caption{Agents in cycles charge themselves with a charging factor or \(1\)}
      \label{fig:cycle-charging}
    \end{subfigure}
    \begin{subfigure}[b]{0.49\textwidth}
      \centering
      \scalebox{0.5}{
      \begin{tikzpicture}
        
        \node at (-1.3,1) {\(a_1\)};
        \node at (0.3,0) {\(c_1\)};
        \node at (1.3,1) {\(a_2\)};
        \node at (2.3,0) {\(c_2\)};
        \node at (3.3,1) {\(a_3\)};
        \node at (4.3,0) {\(c_3\)};

        \node at (6.5,0) {\(c_{k-2}\)};
        \node at (7.6,1) {\(a_{k-1}\)};
        \node at (8.6,0) {\(c_{k-1}\)};
        \node at (9.3,1) {\(a_k\)};
        \draw[very thick,red] (0,0) -- (1,1);
        \draw[very thick,red] (2,0) -- (3,1);
        \draw[very thick,red] (6,0) -- (7,1);
        \draw[very thick,red] (8,0) -- (9,1);
        
        \draw[very thick,blue] (1,1) -- (2,0);
        \draw[very thick,blue] (0,0) -- (-1,1);
        \draw[very thick,blue] (3,1) -- (4,0);
        \draw[very thick,blue] (7,1) -- (8,0);

        \path[->,>={Stealth[scale=1.2]}, every loop/.style={min distance=2cm, looseness=40}, thin] (1,1)  edge  [in=135,out=45,loop] node [label=below:{1}] {} (1,1);
        \path[->,>={Stealth[scale=1.2]}, every loop/.style={min distance=2cm, looseness=40}, thin] (3,1)  edge  [in=135,out=45,loop] node [label=below:{1}] {} (3,1);
        \path[->,>={Stealth[scale=1.2]}, every loop/.style={min distance=2cm, looseness=40}, thin] (7,1)  edge  [in=135,out=45,loop] node [label=below:{1}] {} (7,1);

        \path[->,>={Stealth[scale=1.3]}, thin] (9,1)  edge  [in=45,out=135] node [label=above right:{\({\alpha_k}/{\alpha_1}\)}] {} (-1,1);

        \draw[fill=black] (-1,1) circle (2pt);
        \draw[fill=black] (0,0) circle (2pt);
        \draw[fill=black] (1,1) circle (2pt);
        \draw[fill=black] (2,0) circle (2pt);
        \draw[fill=black] (3,1) circle (2pt);
        \draw[fill=black] (4,0) circle (2pt);
        \draw[fill=black] (4.5,0.5) circle (1pt);
        \draw[fill=black] (5,0.5) circle (1pt);
        \draw[fill=black] (5.5,0.5) circle (1pt);
        \draw[fill=black] (6,0) circle (2pt);
        \draw[fill=black] (7,1) circle (2pt);
        \draw[fill=black] (8,0) circle (2pt);
        \draw[fill=black] (9,1) circle (2pt);
        
      \end{tikzpicture}}
      \caption{Agents who are matched in both \(N_i\) and \(M_i\) charge themselves. Agent~\(a_k\) charges \(a_1\) with a factor or \({\alpha_k}/{\alpha_1}\). Red edges represent \(N_i\) and blue edges represent \(M_i\)}
      \label{fig:path-charging}
    
    \end{subfigure}
    \caption{Charging schemes}
    \end{figure}
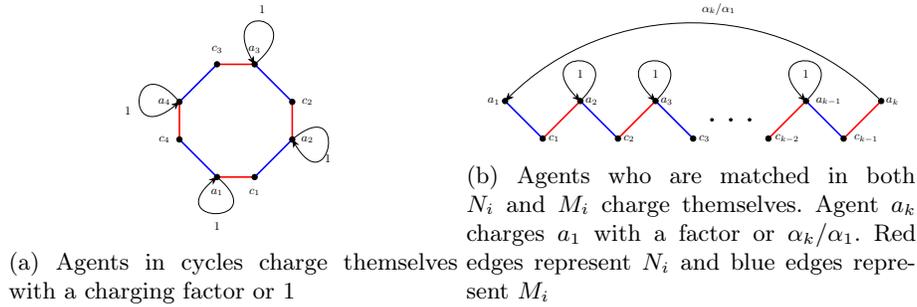

    Now, Consider a component which is a path \(\rho\). There are two cases.
    
    \begin{enumerate}
    \item {\em Case 1: The path $\rho$ has an even length: }   If $\rho$ starts and ends at a category node, then each agent along the path is matched  in both \(N_i\) and \(M_i\). Hence, all such agents can charge themselves with a charging factor of \(1\).
    Suppose $\rho$ starts and ends at an agent as shown in Fig~\ref{fig:path-charging} i.e. \(\rho=(a_1,c_1,a_2,c_2,\cdots,a_{k-1},c_{k-1},a_k)\). Let \(a_1\) be matched in \(M_i\) and \(a_k\) is matched in \(N_i\). Then, \(\alpha_1\) must be greater than or equal to \(\alpha_k\). Otherwise from Lemma~\ref{lemma:max-size-max-wt}, \(M_i \oplus \rho\) is a matching of higher weight - which contradicts the fact that \(M_i\) is the maximum weight matching. Now, every agent in \(\rho\) except \(a_1\) and \(a_k\) charge themselves with a charging factor of \(1\) and \(a_k\) charges \(a_1\) with a charging factor of \({\alpha_k}/{\alpha_1}\).
    \item {\em Case 2: The path $\rho$ has an odd length:} Then either $\rho$ begins and ends with an $M_i$ edge or with an $N_i$ edge. If $\rho$ starts and ends with an $M_i$ edge, then every agent along the path who is matched in \(N_i\) is also matched in \(M_i\). Therefore all the agents on $\rho$ charge themselves.

Consider the case when $\rho$ starts with an \(N_i\) edge. Since $c_k$ is an end-point of $\rho$ with an $N_i$-edge, $c_k$ must have more agents matched to it in $N_i$ than that in $M_i$. So $c_k$ cannot be saturated in $M_i$.

    
    As \(M_i\) is a maximum size matching [\ref{lemma:max-size-max-wt}], we cannot augment $M_i$ to $M_i\oplus \rho$ in $G_i$ even though both endpoints are unsaturated. This can happen only because the daily supply is met. That is $|M_i| = s_i$.  As $a_1$ is vaccinated in category $c_1$ in $N_i$, we claim that the weight $w(a_1, c_1)$ is less than every other edge in $M_i$. This is because if there exists an edge $e \in M_i$ such that $w(e) < w(a_1, c_1)$, we can remove the edge $e$ from $M_i$ and apply the augmenting path $\rho$ to get a matching with a higher weight, which is a contradiction. Therefore, as $w(a_1, c_1)$ is less than every other edge in \(M_i\), agent \(a_1\) can safely charge any agent \(a_q\) who is matched in \(M_i\). Since \(|M_i| \ge |N_i|\), we are guaranteed to have sufficient agents in \(N_i\) for charging.  
    \end{enumerate}
  \end{proof}

\emph{Order of charging among Type~2 agents:} First, every agent who has both $M_i$ and $N_i$ edges indecent on it, charges herself. Next every agent who is an end-point of an even-length path charges the agent represented by the other end-point. The rest of the agents are end-points of an odd-length path matched in $N_i$. We proved that the edges incident on these agents have a weight smaller than every edge in $M_i$. They can charge any agent of $M_i$ who has not been charged yet by any agent of $N_i$, as stated above.

\begin{proof}[of Theorem~\ref{thm:gen-on-same-utility}~(i)]
  Let $a_q$ be an agent who is vaccinated by the online matching \(M\) on day $i$. Then $a_q$ can be charged by at most two agents matched in \(N\). Suppose $a_q$ is vaccinated by the optimal matching \(N\) on some day $i' >i$. Assume that the agent $a_p$ of type 2 who also charges $a_q$. If the priority factor of $a_q$ and $a_p$ are $\alpha_q$ and $\alpha_p$ respectively, then   
  \begin{align*}
    &\frac{\alpha_p.\delta^i + \alpha_q.\delta^{i'}}{\alpha_q.\delta^i}\  =\  \left(\frac{\alpha_p}{\alpha_q}\right)^i + {\delta^{i'-i}} \leq\  {1 + \delta}. 
  \end{align*}
The last inequality follows as \(0 < \alpha_p \leq \alpha_q < 1,  \text{ and }i' > i\).  Therefore the utility obtained by $a_p$ and $a_q$ in $M_i$ is atmost $1 + \delta$ times the the utility of $a_q$ in $M_i$. Therefore the competitive ratio of Algorithm~\ref{alg:online-greedy-m1} is at most ${1 + \delta}$.  
\end{proof}

In the Appendix, we show a tight example which achieves this compititve ratio.

Since the daily supply of day \(d_1\) is \(1\), vaccinating \(a_1\) maximizes the utility gained on the first day. Hence there exists a run of Algorithm~\ref{alg:online-greedy-m1} where \(a_1\) is  vaccinated under category \(c_1\) on day \(d_1\). In this run, agent \(a_2\) cannot be vaccinated on day \(d_2\) as she is unavailable on that day. Hence, total utility gained by the online allocation is \(\alpha_1\). Whereas in a optimal allocation scheme all the agents can be vaccinated. We vaccinate agent \(a_2\) on day \(d_1\) under category \(c_2\), agent \(a_1\) on day \(d_2\) under category \(c_1\). This sums to a total utility of \(\alpha_1+\alpha_1\delta\). Therefore the competitive ratio is $\frac{\alpha_1+\alpha_1\delta}{\alpha_1} = 1 + \delta$. 








\subsection{Tight example for the Online Algorithm}
\begin{figure}[!h]
  \centering
  \scalebox{1}{
  \begin{tikzpicture}[
    on_node/.style={circle, draw=red!30, fill=red!20, ultra thin, minimum size=4mm},
    opt_node/.style={circle, draw=green!30, fill=green!20, ultra thin, minimum size=4mm},
    null_node/.style={circle, draw=white!10, fill=white!10, ultra thin, minimum size=4mm},
    node distance=0.2cm,->,>=stealth',
    ]
    \node (a1) {\(a_1\)};

    \node[on_node] (a1d1) [right=of a1] {$\alpha_1$};
    \node[opt_node] (a1d2) [right=1cm of a1d1] {\(\alpha_1\delta\)};

    \node[opt_node] (a3d1) [below=of a1d1] {$\alpha_1$}; 
    \node[null_node] (a3d2) [below=of a1d2] {-};

    \node (a3) [left=of a3d1] {\(a_2\)};

    \node (d1) [above=of a1d1] {\(d_1\)};
    \node (d2) [above=of a1d2] {\(d_2\)};

    \draw[->] (a1d2.west) -- (a1d1.east);
    \draw[->] (a3d1) to [out=30,in=310] (a1d1);
  \end{tikzpicture}}
  \caption[Tight Example]{A tight example with competitive ratio \(1 + \delta \). Online allocation indicated in red, Optimal allocation indicated in green and arrows indicate charging} 
  \label{fig:tight-model-1}
\end{figure}
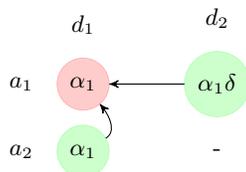

The following example shows that the competitive ratio of Algorithm~\ref{alg:online-greedy-m1} is tight. Let the set of agents \(A=\{a_1, a_2\}\) and categories \(C=\{c_1, c_2\}\). Agent \(a_1\) is eligible under \(\{c_1,c_2\}\) and agent \(a_2\) is eligible only under \(\{c_2\}\).  The daily~supply: \(s_{1}=1 \text{ and } s_{2}=1\). The daily quota of each category on each day is set to \(1\). The  priority factor for both the agents  is~\(\alpha_1\). Assume that \(a_1\) is available on both the days whereas the agent \(a_2\) is available only on the first day.
Figure~\ref{fig:tight-model-1} depicts this example. 

\section{Online Algorithm for Model~$2$}\label{sec:onlineapprox}
We present an online algorithm which greedily maximizes utility on each day. We assume that the discounting factor of the agents is \(\delta\). Moreover each agent $a_k$ has a priority factor $\alpha_k$. Let $\alpha_{\max} = \max_i\{\alpha_i \mid \alpha_i \text{ is the priority factor of agent $a_i$}\}$ and $\alpha_{\min} = \min_i\{\alpha_i \mid \alpha_i \text{ is the priority factor of agent $a_i$}\}$. We show that this algorithm indeed achieves a competitive ratio of \(1+\delta + \frac{\alpha_{\max}}{\alpha_{\min}}\delta\). 

\emph{Outline of the Algorithm:} On each day \(d_i\), starting from day \(d_1\), we construct a bipartite graph \({H}_i=({A}_i\cup C, E_i, w_i)\) where set \({A}_i\) is the set of agents who are available on day \(d_i\) and are not vaccinated earlier than day $d_i$.  Let the weight of the
edge $(a_j , c_k) \in E_i$ be $w_i(a_j , c_k) = \alpha_j .\delta^{i-1}$. Let \(b'_{i,k}\) represent the capacity of \(c_k\in C\) in \(H_i\). In this graph, our algorithm finds a maximum weighted b-matching of size not more than the daily supply value \(s_i\). This can be found in polynomial time \cite{lawler2001combinatorial}. Lemma~\ref{lemma:max-size-max-wt} proves that the maximum weight b-matching is also a maximum cardinality b-matching of $H_i$.  

\begin{algorithm}
\caption{Online Algorithm for Vaccine Allocation}
\label{alg:online-greedy-m2}
\hspace*{\algorithmicindent} \textbf{Input:} An instance \(I\) of Model~2 \\
\hspace*{\algorithmicindent} \textbf{Output:} An allocation \(M:A\to(C\times D)\cup\{\varnothing\}\)

\begin{algorithmic}[1]
  \STATE Let \(D,A,C\) be the set of Days, Agents and Categories respectively.
  \STATE \(M(a_j)\gets \varnothing\) for each agent $a_j \in A$
  \STATE \(r_k \gets q_k\) for each category $c_k \in C$
  \FOR {day \(d_i\) in \(D\)}
  \STATE ${A}_i \gets \{ a_j \in A \mid a_j$ is available on $d_i$ and $a_j$ is not vaccinated\}
  \STATE $E_i=\{(a_j,c_k) \in {A}_i \times C \mid a_j$ is eligible to be vaccinated under category $c_k \}$
  \STATE Construct bipartite graph \({H}_i=({A}_i\cup C, E_i)\).
  \FOR {\(c_k\) in \(C\)}
  \STATE \(b'_{i,k} \gets min(q_{ik}, r_k)\) \COMMENT{Capacity for each \(c_k\) in \(H_i\)}
  \ENDFOR
    \STATE Find maximum weight b-matching \(N_i\)  in \(H_i\) of size at most \(s_i\).
  \FOR {each edge \((a_j,c_k)\) in \(M_i\)}
  \STATE \(M(a_j) \gets (c_k,d_i)\) \COMMENT{Mark \(a_j\) as vaccinated on day \(d_i\) under category \(c_k\)}
  \STATE \(r_k \gets r_k -1 \) \COMMENT {Update remaining overall quota}
  \ENDFOR
  \ENDFOR\\
  \RETURN \(M\)
\end{algorithmic}
\end{algorithm}

\subsection{Outline of the charging scheme}
We compare the solution obtained by Algorithm~\ref{alg:online-greedy-m2} with the optimal offline solution to get the worst-case competitive ratio for Algorithm~\ref{alg:online-greedy-m2}. Let $M$ be the output of Algorithm~\ref{alg:online-greedy-m2} and $N$ be an optimal offline solution. To compare $M$ and $N$, we devise a {\em charging scheme} similar to that in Section~\ref{sec:analysis-model-1}, by which each agent $a$ matched in $N$ {\em charges} a unique agent $a'$ matched in $M$. The amount charged, referred to as the {\em charging factor} here is the ratio of utilities obtained by matching $a$ and $a'$ in $M$ and $N$ respectively.

Properties of the charging scheme:
\begin{enumerate}
    \item Each agent matched in $N$ charges exactly one agent matched in $M$,
    \item Each agent matched in $M$ is charged by at most three agents matched in $N$, with charging factors at most $1,\delta$ and $\frac{\alpha_{\max}}{\alpha_{\min}}\delta$. This implies that the utility of $N$ is at most $(1+\delta+\frac{\alpha_{\max}}{\alpha_{\min}}\delta)$ times the utility of $M$.
\end{enumerate}

We divide the agents matched in $N$ into two types. Type $1$ agents are those which are matched in $M$ on an earlier day compared to that in $N$. Thus $a\in A$ is a Type $1$ agent if $a$ is matched on day $d_i$ in $M$ and on day $d_j$ in $N$, such that $i<j$. The remaining agents are called Type $2$ agents.
Our charging scheme is as follows:
    
\begin{enumerate}
    \item Type $1$ agents charge themselves with a charging factor $\delta$, since the utility associated with them in $N$ is at most $\delta$ times that in $M$. 
    
    \item Here onwards, we consider only Type $2$ agents and discuss the charging scheme associated with them.

Let $X_i$ be the set of Type $2$ agents matched on day $d_i$ in $N$, and let $Y_i$ be the set of agents matched on day $d_i$ in $M$. 

\begin{enumerate}
    \item {\em Case 1: $|X_i|\leq |Y_i|$: } From Lemma \ref{lemma:injection-exists} we claim that each agent $a_p \in X_i$ charges an agent in $a_q \in Y_i$ with $\alpha_p \leq \alpha_q$. Therefore the agents in $X_i$ charge the agents in $Y_i$ with a charging factor of $1$.

\item {\em Case 2: $|X_i|=|Y_i|+z, z>0$: } Let $N_i$ and $M_i$ respectively be the restrictions of $N$ and $M$ to day $d_i$. We construct an auxiliary bipartite graph $G_i$ where $X_i\cup Y_i$ form one bipartition and categories form another bipartition. The edge set is $N_i\cup M_i$. For a category $c_k$, let $n_{j,k}$ and $m_{j,k}$ be the number of agents matched in $N$ and $M$ respectively, under category $c_k$ on day $d_j$. Then we set the quota of $c_k$ in $G_i$ to be $b_{i,k}=\min\{q_{i,k}, \max\{q_k-\sum_{j=1}^{i-1}n_{j,k}, q_k-\sum_{j=1}^{i-1} m_{j,k}\}\}$. This is the maximum of the quotas of $c_k$ that were available for computation of $N_i$ and $M_i$ respectively.
        
The charging scheme is given by the following. Consider the symmetric difference $M_i\oplus N_i$. Since $|N_i|=|M_i|+z$, there are exactly $z$ edge-disjoint alternating paths in $M_i\oplus N_i$ that start and end with an edge of $N$ \cite{NasreR17}. Let $\rho=\langle a_1,c_1,a_2,\ldots, a_k, c_k\rangle$ be one such path. Then $a_2,\ldots,a_{k-1}$ are matched in both $M_i$ and $N_i$, so they charge themselves with a charging factor of $1$. From Lemma \ref{lemma:injection-exists}, the agent $a_1$ charges $a_k$ with charging factor of at most $1$. It remains to decide whom $a_k$ charges.
        
Since $\rho$ terminates at $c_k$ with an $N_i$-edge, the number of agents matched to $c_k$ in $N_i$ is more than those matched to $c_k$ in $M_i$. In Lemma~\ref{lem:saturated}, we show that this can happen only because of exhaustion of $q_k$ in Algorithm~\ref{alg:online-greedy-m2} on or before day $d_i$. 
So agent $a_k$ can charge some agent $a_l$ matched to $c_k$ in $M$ on an earlier day, with charging factor $\frac{\alpha_{k}}{\alpha_{l}}\delta \leq \frac{\alpha_{\max}}{\alpha_{\min}}\delta$.
\end{enumerate}
\end{enumerate}

\begin{lemma}\label{lem:saturated}
If node $c_k$ is an end-point of a path $\rho$ in $G_i$, then $q_k$ is exhausted in Algorithm \ref{alg:online-greedy-m2} on or before day $d_i$. 
\end{lemma}
\begin{proof}
Suppose $c_k$ be an endpoint of $\rho$ in $G_i$. The number of agents matched to $c_k$ in $N_i$ is more than those matched to $c_k$ in $M_i$. We know that the daily supply $s_i$ of the day $d_i$ is an upperbound for both $|M_i|$ and $|N_i|$. Since $|N_i| = |M_i| + z$, we have $|M_i| < s_i$. From Algorithm \ref{alg:online-greedy-m2} we know that $M_i$ is a maximum-size b-matching in $H_i$ of size at most \(s_i\). If the capacity of $c_k$ is not saturated in $H_i$, then we can augment the path $\rho$ contradicting the maximality of $M_i$.  Since $c_k$ has more edges of $M_i$ than $N_i$ incident to it, from the definition of $b_{i,k}$, category $c_k$ must have exhausted the overall quota $q_k$ in Algorithm \ref{alg:online-greedy-m2} on or before day $d_i$. 
\end{proof}

\subsection {Tight Example}
The following example shows that the competitive ratio of Algorithm~\ref{alg:online-greedy-m2} is tight. Let set of agents \(A=\{a_1, a_2, a_3\}\) and categories \(C=\{c_1, c_2\}\). Agent \(a_1\) is eligible under \(\{c_1,c_2\}\). Agent \(a_2\) is eligible only under \(\{c_1\}\) and agent \(a_3\) is eligible only under \(\{c_2\}\).  The daily~supply: \(s_{1}=1 \text{ and } s_{2}=2\). Overall~quotas: \(q_1=1 \text{ and } q_2=2\). The daily quota of each category on each day is set to \(1\). The utility discounting factor for each agent is~\(\delta\).  The priority factor of the agent $a_i$ is $\alpha_i$ for $i = 1,2,3$. We assume that $0 \leq \alpha_1 = \alpha_3 < \alpha_2 \leq 1$. Agent \(a_1\) is available on both the days.  Agent \(a_3\) is available only on the first day, whereas agent $a_2$ is available only on the second day.  
Figure~\ref{fig:tight-general} depicts this example. 
\begin{figure}
  \centering
  \begin{tikzpicture}[
    on_node/.style={circle, draw=red!30, fill=red!20, ultra thin, minimum size=4mm},
    opt_node/.style={circle, draw=green!30, fill=green!20, ultra thin, minimum size=4mm},
    null_node/.style={circle, draw=white!10, fill=white!10, ultra thin, minimum size=4mm},
    node distance=0.2cm,->,>=stealth',
    ]
    \node (a1) {\(a_1\)};
    \node (a2) [below=1cm of a1] {\(a_2\)};
    \node (a3) [below=1cm of a2] {\(a_3\)};

    \node[on_node] (a1d1) [right=of a1] {$\alpha_1$};
    \node[opt_node] (a1d2) [right=1cm of a1d1] {\(\alpha_1\delta\)};
    
    \node[null_node] (a2d1) [right=of a2] {-};
    \node[opt_node] (a2d2) [right=1.35cm of a2d1] {\(\alpha_2\delta\)};
    
    \node[opt_node] (a3d1) [right= of a3] {$\alpha_3$}; 
    \node[null_node] (a3d2) [right=1.35cm of a3d1] {-};

    \node (d1) [above=of a1d1] {\(d_1\)};
    \node (d2) [above=of a1d2] {\(d_2\)};

    \draw[->] (a1d2.west) -- (a1d1.east);
    \draw[->] (a2d2) -- (a1d1);
    \draw[->] (a3d1) to [out=30,in=310] (a1d1);
  \end{tikzpicture}
  \caption[Tight Example]{A tight example with competitive ratio \(1+\delta+\frac{\alpha_2}{\alpha_1}\delta\). Online allocation indicated in red, Optimal allocation indicated in green and arrows indicate charging} 
  \label{fig:tight-general}
\end{figure}
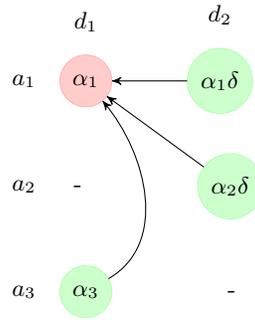

Since the daily supply of day \(d_1\) is \(1\), vaccinating \(a_1\) maximizes the utility gained on the first day. Hence there exists a run of Algorithm~\ref{alg:online-greedy-m2} where \(a_1\) is  vaccinated under category \(c_1\) on day \(d_1\). In this run, agent \(a_2\) cannot be vaccinated on day \(d_2\) as she is eligible only under category \(c_1\) and overall quota of category \(c_1\) is exhausted. Hence, total utility gained by the online allocation is \(\alpha_1\). Whereas in a optimal allocation scheme all the agents can be vaccinated. Vaccinate agent \(a_3\) on day \(d_1\) under category \(c_2\), agent \(a_1\) and \(a_2\) on day \(d_2\) under categories \(c_2,c_1\) respectively. This sums to a total utility of \(\alpha_3+\alpha_1\delta+{\alpha_2}\delta\). Therefore the competitive ratio of the online algorithm is $\frac{\alpha_3+\alpha_1\delta+{\alpha_2}\delta}{\alpha_1} = \frac{\alpha_1+\alpha_1\delta+{\alpha_2}\delta}{\alpha_1} = 1+\delta+\frac{\alpha_{\max}}{\alpha_{\min}}\delta$. The first equality holds as $\alpha_1 = \alpha_3$. The second equality holds as $\alpha_{\max} = \alpha_2$ and $\alpha_{\min} = \alpha_1$.

 \section{Strategy-proofness of the online algorithm}
We give the details of the Pure Nash Equilibrium here.


\subsection{Pure Nash equilibrium}\label{sec:Nash-eq}
The offline algorithm might choose any arbitrary matching that maximizes the utility. We present a deterministic tie-breaking rule similar to the one used in \cite{Aziz21} to force the algorithm to pick a unique matching. For this, we fix an ordering \(\pi\) on agents. We show the existence of a pure Nash equilibrium under the deterministic tie-breaking. We cast our problem as a linear program as given in Fig~\ref{fig:lp1}.

\begin{figure}
\begin{alignat*}{2}
  & \text{maximize: } & & \smashoperator[l]{\sum_{\substack{i\in A, j\in C,\\ k\in D}}} u_{ik}.x_{ijk} \\
    & \text{subject to: }& \quad & \smashoperator[l]{\sum_{i\in A,j\in C}}
                    \begin{aligned}[t]
                        x_{ijk} & \le s_k,& \forall k & \in D\\[3ex]
                    \end{aligned}\\
    &             & \quad & \smashoperator[l]{\sum_{i\in A\ \ \ \ }}
                    \begin{aligned}[t]
                        x_{ijk} & \le q_{jk},& \forall (j,k) & \in C\times D\\[3ex]
                    \end{aligned}\\
    &             & \quad & \smashoperator[l]{\sum_{j\in C,k\in D}}
                    \begin{aligned}[t]
                        x_{ijk} & \in [0,1],& \forall i & \in A\\[3ex]
                    \end{aligned}\\
    &             & \quad & \quad\ \ \ \ \ \ \ \
                    \begin{aligned}[t]
                        x_{ijk} & \in [0,1],& \forall (i,j,k) & \in A\times C \times D\\[3ex]
                    \end{aligned}\\
\end{alignat*}
\caption{Here $u_{ik}$ is the utility value of agent $i$ on day $k$, and $ s_k \& q_{jk}$ are the daily supply and daily quotas respectively.}
\label{fig:lp1}
\end{figure}

 It can be seen that this LP models the network flow formulation of our problem stated in Section~\ref{sec:offline-optimal-algorithm}. It is known (\cite{lawler2001combinatorial}) that the polytope arising from the network flow problem is integral. To impose the deterministic tie breaking, we modify the objective function as follows. 
 
 \begin{align*}
  \text{maximize\ }& \smashoperator[l]{\sum_{\substack{i\in A, j\in C,\\ k\in D}}} u_{ik}.x_{ijk} + \lambda \times REG, \text{where}\\
  REG &= \sum_{i \in A} \frac{\sum_{k \in D, j \in C} x_{ijk}}{2^{\pi(i)}}
\end{align*}

For a sufficiently small \(\lambda\) \((\lambda < \delta^{|D|+1})\), the difference between utilities of any two allocations is greater than \(REG\). Therefore, the linear program in Figure~\ref{fig:lp1} maximizes the objective function in Fig~\ref{fig:lp1}, but breaks ties to maximize \(REG\).

Let \(A_{d_i}\) be defined as the set of agents matched on a day \(d_i\in D\) and \(A_\infty\) be the set of unmatched agents at the end of a run of the Algorithm~\ref{alg:online-greedy-m1}. Let agent \(a_p\) be matched on \(d_i\) (WLOG, assume all unmatched agents are matched on day \(\infty\). Now, we present a proof of Theorem~\ref{thm:nash}. 

\begin{proof}(of Theorem~\ref{thm:nash})
Suppose the agent \(a_p\) is matched on day $d_i$, and deviates to reporting a subset of the actual available days. 

If agent \(a_p\) gets matched on a day \(d_j\), $j<i$, because of misreporting her available days, then some agent \(a_q\) on day \(d_j\) will remain unmatched. This follows, since on any given day, the matching  computed by algorithm \ref{alg:online-greedy-m1} is of maximum size and all agents other than \(a_p\) turn up on at most one day. The rest of the matching will remain unchanged.
But, agent \(a_q\) is prioritized by \(\pi\) over agent \(a_p\). Otherwise, algorithm \ref{alg:online-greedy-m1} would have matched \(a_p\) and not \(a_q\) on day \(d_j\). Hence, agent \(a_p\) cannot replace agent \(a_q\) on day \(d_j\) even after misreporting her availability.

Therefore agent \(a_p\) has no advantage in deviating from the strategy. Hence, the above matching is a pure Nash equilibrium.
\end{proof}


\section{Experimental Evaluation}\label{sec:simulation}
In Section~\ref{sec:model1} we prove worst-case guarantees for the online algorithm. We also give a tight example instance achieving a competitive ratio of \(1+2\delta\). Here, we experimentally evaluate the  performance of the online algorithm and compare it with the worst-case guarantees on a real-life dataset. For finding the optimal allocation that maximizes utility, we solve the  networkflow linear program  with the additional constraint for overall quota $\sum_{i\in A, k\in D} x_{ijk}\leq q_j\quad \forall c_j\in C$. This LP is described in the Appendix. The code and datasets for the experiments can be found at~\cite{exp-github}

\subsection{Methodology}
All experiments run on a 64-bit Ubuntu 20.04 desktop of 2.10GHz * 4 Intel Core i3 CPU with 8GB memory. 

The proposed online approximation algorithm runs in polynomial time. In contrast, the optimal offline algorithm solves an integer linear program which might take exponential time depending on the integrality of the polytope. We relax the integrality constraints to achieve an upper-bound on the optimal allocation. For comparing the performance of the online Algorithm~\ref{alg:online-greedy-m1} and the offline Algorithm, we use vaccination data of 24 hospitals in Chennai, India for the month of May 2022. We use small data-sets with varying instance sizes for evaluating the running times of the algorithms. We use large data-sets of smaller instance sizes for evaluating competitive ratios.      

All the programs used for the simulation are written in Python language. For solving LP, ILP, and LPR, we use the general mathematical programming solver COIN-OR Branch and Cut solver MILP (Version:~2.10.3)\cite{CBC} on PuLP (Version~2.6) framework\cite{pulp}. When measuring the running time, we consider the time taken to solve the LP. 

\subsection{Datasets}
Our dataset can be divided into two parts.

\textit{Supply:} We consider vaccination data of twenty four hospitals of Chennai, India for the month of May 2022. This data is obtained from the official COVID portal of India using the API's provided. The data-set consists of details such as daily vaccination availability, type of vaccines, age limit, hospital ID, hospital zip code, etc. for each hospital. 

\textit{Demand:} Using the Google Maps API \cite{google-maps-api}, we consider the road network for these 24 hospitals in our data-set. From this data we construct a complete graph with hospitals as vertices and edge weights as the shortest distance between any two hospitals. For each hospital \(h\in H\), we consider the cluster \(C(h)\) as the set of hospitals which are at most five kilo meters away from \(h\). We consider these clusters as our categories. Now, we consider 10000 agents who are to be vaccinated. For each agent \(a\), we pick a hospital \(h\) uniformly at random. The agent \(a\) belongs to every hospital in the cluster \(C(h)\). Each agent's availability over 30 days is independently sampled from the uniform distribution. Now, we consider the age wise population distribution of the city. For each agent we assign an age sampled from this distribution. Now, we partition the set of agents as agents of age 18-45years, 45-60years and 60+. We assign \(\alpha\)-values \(0.96,0.97 \text{ and } 0.99\) respectively. We also consider the same dataset with \(\alpha\)-values \(0.1,0.5 \text{ and } 0.9\) respectively. We set the discounting factor \(\delta\) to be 0.95.

For analyzing the running time of our algorithms, we use synthetically generated datasets with varying number of instance sizes ranging from 100 agents to 20000 agents. Each agent's availability and categories are chosen randomly from a uniform distribution.

\subsection{Results and Discussions}

We show that the online algorithm runs significantly faster than the offline algorithm while achieving almost similar results. We give a detailed emperical evaluation of the running times in the Appendix. 

\ \\
To compare the performance of the online Algorithm~\ref{alg:online-greedy-m1} against the offline algorithm we define a notion of \textit{remaining fraction of un-vaccinated agents}. That is, on a given day \(d_i\), we take the set of agents \(P_{d_i}\) who satisfy both of the following conditions:
\begin{enumerate}
\item Agent \(a\) is available on some day \(d_j\) on or before day \(d_i\).
\item Agent \(a\) belongs to some hospital \(h\) and \(h\) has non-zero capacity on day \(d_j\)
\end{enumerate}

\(P_{d_i}\) is the set of agents who could have been vaccinated without violating any constraints. Let \(\gamma_{i} = \lvert P_{d_i} \rvert\).

Let \(V_{d_i}\) be the set of agents who are vaccinated by the algorithm on or before day \(d_i\).  Let \(\eta_i = \lvert V_{d_i} \rvert\). Now, \(1-\eta_i / \gamma_i \) represents the fraction of unvaccinated agents. In Figure~\ref{fig:final_compare} we compare the age-wise  \(1-\eta_i / \gamma_i \) of both of our online and offline algorithms. We note that the vaccination priorities given to vulnerable groups by the online approximation algorithm is very close to that of the offline optimal algorithm. In both the algorithms, By the end of day 2, 50\% of \(1-\eta_i / \gamma_i \) was achieved for agents of 60+ age group. By the end of day 8, only 10\% of the most vulnerable group remained unvaccinated.

\begin{figure}
    \centering
\includegraphics[width=0.9\textwidth]{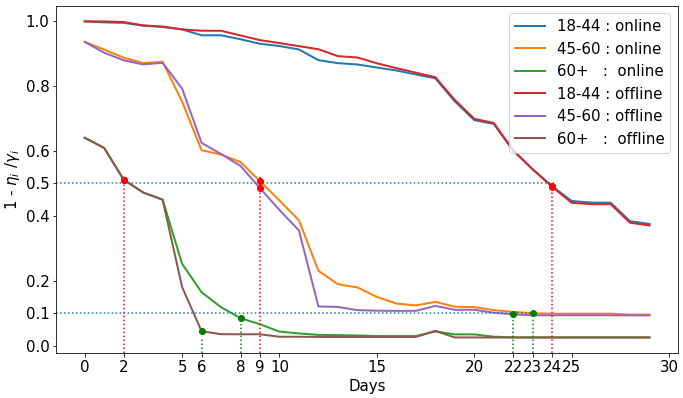}
    \caption{The \(1-\eta_i / \gamma_i \) value achieved by the online algorithm is very similar to that of the offline algorithm across age groups. Both algorithm vaccinate achieves vaccinate 90\% of the most vulnerable group within 8 days. }
    \label{fig:final_compare}
  \end{figure}

\subsection{Running Time Analysis}
In Table~\ref{table:final-table} we compare the performance of the online algorithm and the offline algorithm against the same dataset. We consider alpha values \((0.96,0.97,0.99)\) and \((0.1,0.5,0.9)\). In both the cases, the online algorithm vaccinates almost the same number of agents as that of the offline while algorithm achieving similar total utility. The competitive ratio is \(0.99\). The online algorithm runs significantly faster than the offline algorithm. 

\begin{minipage}{\textwidth}
\begin{minipage}{0.48\textwidth}
\scalebox{0.7}{
\begin{tabular}{|c|ll|ll|}
\hline
                                                                            & \multicolumn{2}{c|}{\begin{tabular}[c]{@{}c@{}}Online\\ Algorithm\end{tabular}} & \multicolumn{2}{c|}{\begin{tabular}[c]{@{}c@{}}Offline\\ Algorithm\end{tabular}} \\ \hline
\(\alpha\) value                                                            & \multicolumn{1}{l|}{\(\vec{\alpha_1}\)}           & \(\vec{\alpha_2}\)          & \multicolumn{1}{l|}{\(\vec{\alpha_1}\)}           & \(\vec{\alpha_2}\)           \\ \hline
\(\delta\)                                                                  & \multicolumn{1}{l|}{0.95}                         & 0.95                        & \multicolumn{1}{l|}{0.95}                         & 0.95                         \\ \hline
\begin{tabular}[c]{@{}c@{}}Running time \\ (in sec)\end{tabular}            & \multicolumn{1}{l|}{319.04}                       & 336.55                      & \multicolumn{1}{l|}{888.90}                       & 806.65                       \\ \hline
\begin{tabular}[c]{@{}c@{}}Total \\ no. of agents\\ vaccinated\end{tabular} & \multicolumn{1}{l|}{7154}                         & 7145                        & \multicolumn{1}{l|}{7192}                         & 7192                         \\ \hline
Total Utility                                                               & \multicolumn{1}{l|}{3567.95}                      & 1550.23                     & \multicolumn{1}{l|}{3580.68}                      & 1573.95                      \\ \hline
\end{tabular}}
\captionof{table}{ The vector \(\vec{\alpha_1} = \)  \((0.96, 0.97, 0.99)\) and vector \(\vec{\alpha_2} = \)  \((0.1, 0.5, 0.9)\) represent the \(alpha\) values for the three age groups . The average competitive ratio is \(0.99\). The average running time of the online and the offline algorithms are 327.79 seconds and 847.77 seconds respectively.  }
\label{table:final-table}
\end{minipage}\hfill
\begin{minipage}[t]{0.48\textwidth}
  \pgfplotstableread[row sep=\\,col sep=&]{
    size   & Off                & On                 \\
    100    & 0.019134283065796  & 0.019458115100861  \\
    500    & 0.138595390319824  & 0.116458749771118  \\
    1000   & 0.228429079055786  & 0.135910940170288  \\
    2000   & 0.592561912536621  & 0.356951093673706  \\
    5000   & 1.8075364112854    & 0.909572839736939  \\
    10000  & 4.43315873146057   & 1.69299368858337   \\
    20000  & 13.7637612819672   &  5.67742729187012  \\
    }\mydata
 \scalebox{0.7}{
	\begin{tikzpicture}
		\begin{axis}[
		            enlarge x limits=0,
		            enlarge y limits=0.1,
		            xbar,
		            bar width=6pt,
		            symbolic y coords={100,500,1000,2000,5000,10000,20000},
		            ytick=data,
		            nodes near coords,
		            every node near coord/.append style={font=\scriptsize},
		            xmax=16,
		            legend style={at={(0.5,1.05), font=\scriptsize},
		            anchor=south,legend columns=-1},
		            ylabel={Size of Instance (number of agents)},
		            xlabel={Running time (in 50 sec)},
		            label style={font=\small},
		            tick label style={font=\small},
		        ]
		        \addplot[fill=blue] table[x=On,y=size]{\mydata};
		        \addplot[fill=red] table[x=Off,y=size]{\mydata};
		        \legend{Online Algorithm, Offline Algorithm}
		    \end{axis}  
	\end{tikzpicture}}
	
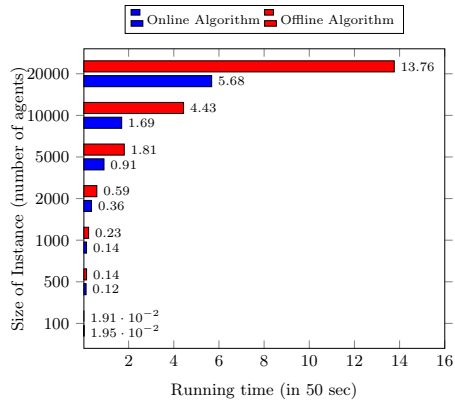
\captionof{figure}{Time taken by offline and online algorithms (on synthetic datasets) vs instance size}
	\label{fig:time_graph}
\end{minipage}
\end{minipage}

  Comparing the running time of Algorithm~\ref{alg:online-greedy-m1} and the offline algorithm,  Figure~\ref{fig:time_graph} shows that the online algorithm runs significantly faster than the offline algorithm for all input sizes.

\subsection{Performance Analysis}
In Figure~\ref{fig:age_0_compare_alphas}, we plot the number of agents of age group 18-45 getting vaccinated by the online algorithm~\ref{alg:online-greedy-m1} on each day for \(alpha\) values \(0.96\)  and \(0.1\). It is clear that the vaccination follows almost identical pattern as long as the order of \(alpha\) values remain the same. Figure~\ref{fig:age_0_compare_alphas_offline} shows similar results for the optimal offline algorithm. The independence on cardinal values shows that the algorithm is practically useful as ordering the vulnerable groups is much more feasible than assigning a particular value.  Similar plots for other age groups are given in the appendix.

  \begin{figure}
    \centering
    \includegraphics[width=1\textwidth]{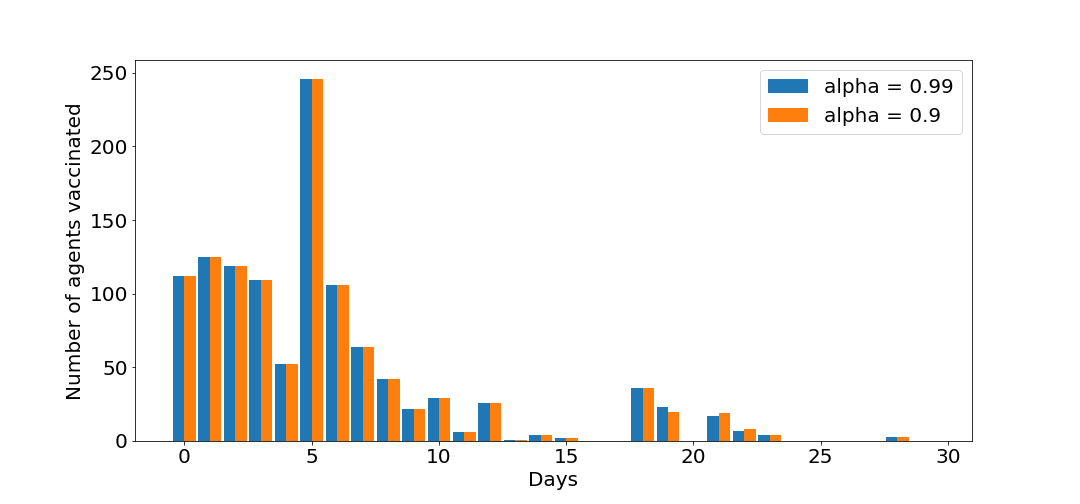}
    \caption{Number of agents in the 60+ age group vaccinated by the online algorithm for \(alpha\)-values \(0.96\) and \(0.1\) respectively.}
    \label{fig:age_0_compare_alphas}
  \end{figure}

  \begin{figure}
    \centering
    \includegraphics[width=1\textwidth]{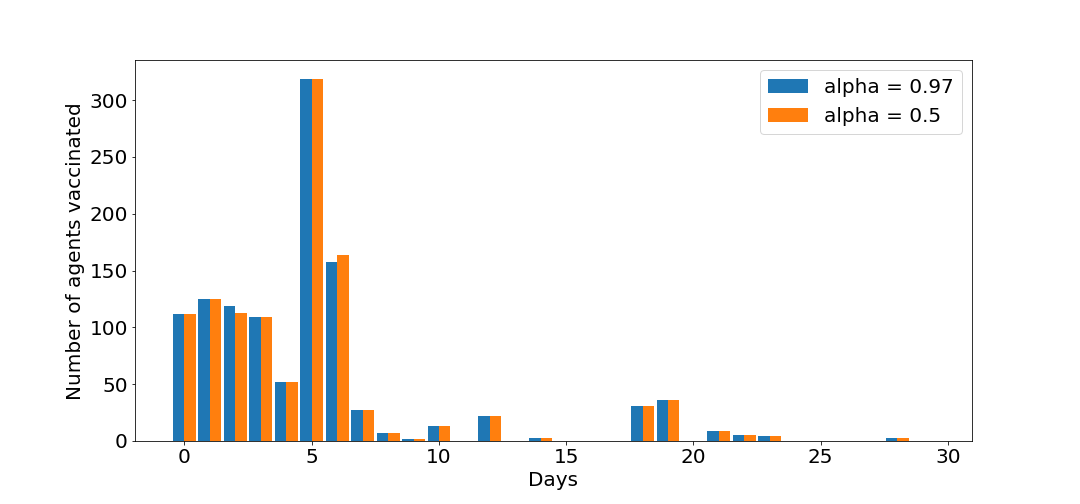}
    \caption{Number of agents in the 60+ age group vaccinated by the offline algorithm for \(alpha\)-values \(0.96\) and \(0.1\) respectively.}
    \label{fig:age_0_compare_alphas_offline}
  \end{figure}

In Figure~\ref{fig:age_1_compare_alphas}, we plot the number of agents of age group 45-60 getting vaccinated by the online algorithm 1 on each day for alpha values 0.97 and 0.5. It is clear that the vaccination follows almost identical pattern as long as the order of alpha values remain the same. Figure~\ref{fig:age_1_compare_alphas_offline} shows similar results for the optimal offline algorithm. Figure~\ref{fig:age_2_compare_alphas} and Figure~\ref{fig:age_2_compare_alphas_offline} plot similar results for the 60+ age group population. We note that in both online and the offline algorithm,  allocations of vaccines for the age group 60+ are higher in the initial days and decreases with days. Most of the agents from this group are vaccinated by the end of 10th day.

  \begin{figure}
    \centering
    \includegraphics[width=1\textwidth]{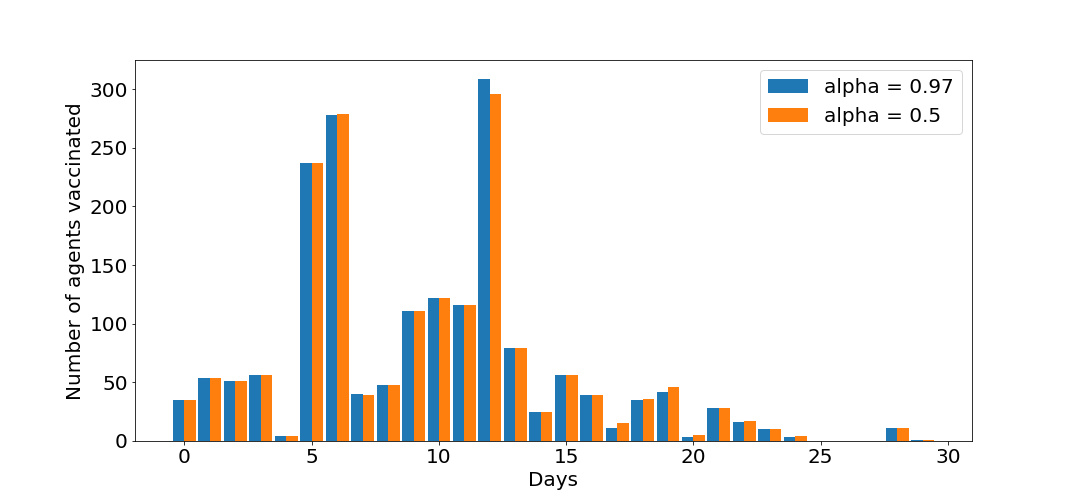}
    \caption{Number of agents in the 45-60 age group vaccinated by the online algorithm for \(alpha\)-values \(0.97\) and \(0.5\) respectively.}
    \label{fig:age_1_compare_alphas}
  \end{figure}

  \begin{figure}
    \centering
    \includegraphics[width=1\textwidth]{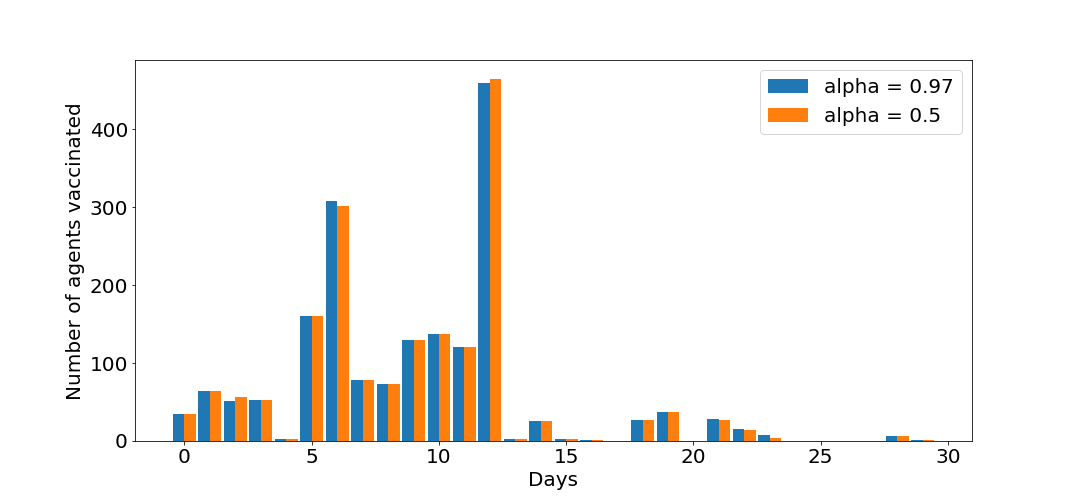}
    \caption{Number of agents in the 45-60 age group vaccinated by the offline algorithm for \(alpha\)-values \(0.97\) and \(0.5\) respectively.}
    \label{fig:age_1_compare_alphas_offline}
  \end{figure}
  
    \begin{figure}
    \centering
    \includegraphics[width=1\textwidth]{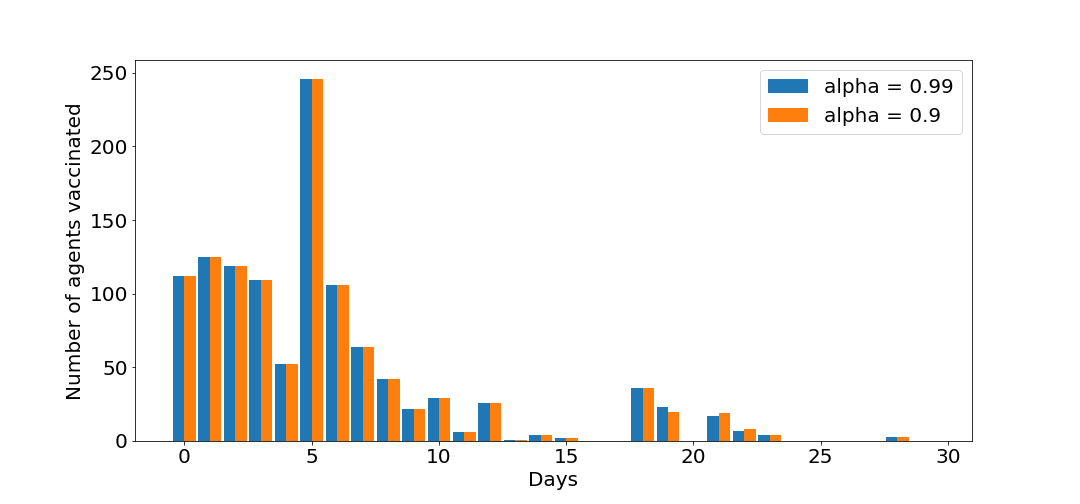}
    \caption{Number of agents in the 60+ age group vaccinated by the online algorithm for \(alpha\)-values \(0.99\) and \(0.9\) respectively.}
    \label{fig:age_2_compare_alphas}
  \end{figure}

  \begin{figure}
    \centering
    \includegraphics[width=1\textwidth]{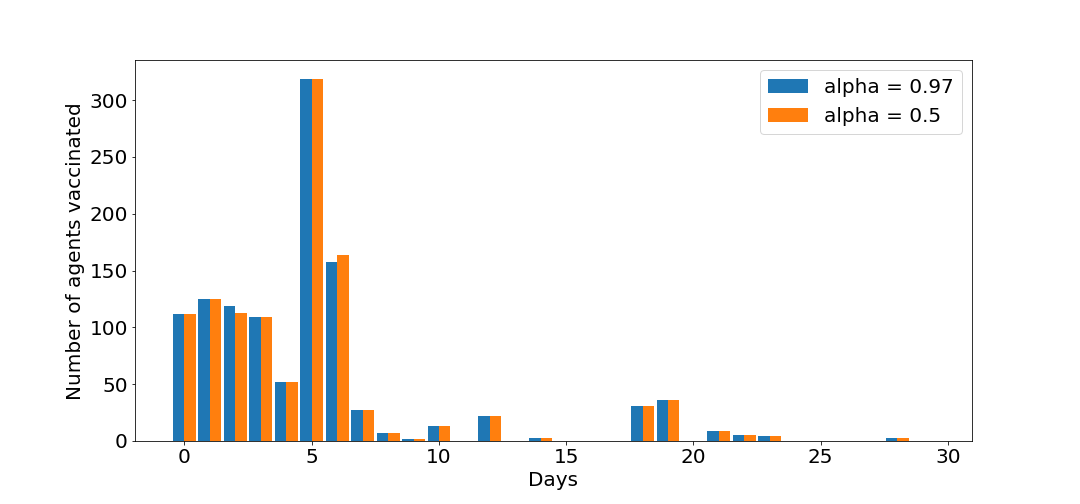}
    \caption{Number of agents in the 60+ age group vaccinated by the offline algorithm for \(alpha\)-values \(0.99\) and \(0.9\) respectively.}
    \label{fig:age_2_compare_alphas_offline}
  \end{figure}

\section{Conclusion}\label{sec:conclusion}
We investigate the problem of dynamically allocating perishable healthcare goods to agents arriving over a period of time. We capture various constraints while allocating a scarce resource to a large population, like production constraint on the resource, infrastructure and constraints. 
While we give an offline optimal algorithm for Model $1$, getting one for Model $2$ or showing NP-hardness remains open.
We also propose an online algorithm approximating welfare that elicits information every day and makes an immediate decision. The online algorithm does not require a foresight and hence has a practical appeal.
Our experiments show that the online algorithm generates a utility roughly equal to the utility of the offline algorithm while achieving very little to no wastage. 
\bibliographystyle{splncs04}
\bibliography{main}
\clearpage


\end{document}